\RequirePackage{fix-cm}
\RequirePackage{amsmath}
\documentclass{svjour3}                     
\smartqed  
\usepackage{subcaption}
\usepackage{hyperref}
\usepackage{graphicx}
\usepackage{amsfonts, amssymb}
\usepackage{enumerate}
\usepackage{natbib}
\usepackage{xcolor}
\newtheorem{observation}{Observation}
\newtheorem{prop}{Proposition}

\newcommand{\V}{{\cal V}}
\newcommand{\N}{{N}}
\newcommand{\I}{{\cal I}}
\newcommand{\leafSet}{{\cal X}}

\newcommand{\rev}[1]{\textcolor{black}{#1}}
\journalname{Journal of Mathematical Biology}
\begin{document}

\title{Distinguishing level-1 phylogenetic networks on the basis of data generated by Markov processes \thanks{Leo van Iersel, Remie Janssen, Mark Jones and Yukihiro Murakami were partly supported by the Netherlands Organization for Scientific Research (NWO), Vidi grant 639.072.602 and Mark Jones also by the gravitation grant NETWORKS.  Elizabeth Gross was supported by the National Science Foundation (NSF), DMS-1945584}
}


\author{Elizabeth Gross \and
        Leo van Iersel  \and
        Remie Janssen   \and
        Mark Jones      \and
        Colby Long      \and
        Yukihiro Murakami
}


\institute{ E. Gross \at
              Department of Mathematics, University of Hawai`i at M\={a}noa,
              2565 McCarthy Mall, Honolulu, Hawai`i 96822 \\
              \email{egross@hawaii.edu}
            \and
            L. van Iersel, R. Janssen, M. Jones, Y. Murakami \at
              Delft Institute of Applied Mathematics, Delft University of Technology, Van Mourik Broekmanweg 6, 2628 XE, Delft, The Netherlands \\
              \email{\{L.J.J.vanIersel, R.Janssen-2, M.E.L.Jones, Y.Murakami\}@tudelft.nl}
            \and
            C. Long \at
              The College of Wooster, 1189 Beall Avenue Wooster, Ohio 44691 \\
              \email{clong@wooster.edu}
}

\date{Received: date / Accepted: date}

\maketitle

\begin{abstract}
Phylogenetic networks can represent evolutionary events that cannot be described by phylogenetic trees. These networks are able to incorporate reticulate evolutionary events such as hybridization, introgression, and lateral gene transfer. Recently, network-based Markov models of DNA sequence evolution have been introduced along with model-based methods for reconstructing phylogenetic networks. For these methods to be consistent,
the network parameter needs to be identifiable from data generated under the model. 
Here, we show that the semi-directed network parameter of a triangle-free, level-1 network model with any fixed number of reticulation vertices is generically identifiable under the Jukes-Cantor, Kimura 2-parameter, or Kimura 3-parameter constraints.
\keywords{phylogenetic networks \and identifiability \and reticulation \and Markov processes}
\end{abstract}

\section{Introduction}

Typically, the goal of a phylogenetic analysis is to find a tree that describes the evolutionary relationships among a set of taxa. However, because trees, as directed graphs, have acyclic skeletons, they cannot represent reticulate
evolutionary events, such as hybridization, introgression, and lateral gene transfer. Recognizing this limitation, 
it has become increasingly common to use phylogenetic networks in order to more accurately describe
the history of some sets of taxa~\citep{bapteste}. 
This increasing attention to phylogenetic networks has led to many
new results about the combinatorial properties of phylogenetic networks \citep{huson2010phylogenetic,gusfield2014recombinatorics},~\citep[Chapter~10]{steel2016phylogeny},
as well as to new methods for inferring phylogenetic networks from biological data. 

Many of these new methods for inferring phylogenetic networks are based on constructing networks from small sets of inferred trees \citep{Baroni2005,Huber2011,RIATA2005,QuartetNet2014} or adapting variants of maximum parsimony and neighbor joining \citep{NeighborNet2004,Jin2007}. 
Several others are model-based methods that are designed to infer various features of a species networks from data generated by a network multispecies coalescent model. These include,
for example, the methods implemented in Phylonet 
\citep{than2008phylonet,  wen2018inferring} as well as SNaQ \citep{Solis2016, solis2017phylonetworks} and NANUQ \citep{allman2019}.
Now that network-based Markov models of DNA sequence evolution have been developed (see e.g. \cite[\S 3.3]{Nakhleh2011}), it seems natural to use these models in order to add other model-based techniques to the set of tools for network inference. 
%
However, in order to consistently infer a parameter using a model-based approach, that parameter must be identifiable from some feature of the model. The question of parameter identifiability is significant and has been explored for several different phylogenetic models. For example, there are numerous identifiability results for tree-based Markov models \citep{Allman2011Identifiability,Allman2006,Chang1996,Rhodes2012} and there are similar results for networks 
that provide the theoretical justification for methods such as SNaQ \citep{solis2020identifiability} and NANUQ \citep{banos2019} mentioned above.  
In this work, we explore the identifiability of the network parameter in network-based Markov models.


Formally, network-based Markov models are parameterized families of probability distributions on $n$-tuples of DNA bases. The parameterization is derived by modeling the process of DNA sequence evolution along an $n$-leaf leaf-labelled topological network, which we call the \emph{network parameter} of the model. 
Given an $n$-taxa sequence alignment, a probability distribution in a network-based Markov model specifies
the probability of observing each of the possible $4^n$ site-patterns  at a particular site.  Indeed, in a model-based approach, an $n$-taxa sequence alignment is usually regarded as an observation of $n$ independent and identically distributed site-patterns. A sequence alignment can therefore be viewed as an approximation of a probability distribution, with the probability for each site-pattern being proportional to the number of times it appears in the alignment.  Given a collection, or class, of network-based Markov models, the network parameter
is \emph{identifiable} if any expected site pattern probability distribution $p$ in the model belongs to at most one model in the class.
Identifiability, as just defined, is very strong and certainly not satisfied for any reasonable collection of 
models.  
Thus, in practice, one often aims at proving that a parameter is \emph{generically identifiable}. If the network parameter of a class of models is generically identifiable then a probability distribution $p$ from one of the models \emph{almost surely} belongs to no other model in the class.

The generic identifiability of
the tree and network parameters of several phylogenetic models has been shown by adopting techniques from algebraic geometry
\citep{Allman2011Identifiability,Gross2017DistinguishingPN,hollering2020Identifiability,Long2017Identifiability}. These results apply to several types of mixture models, network models, and multispecies coalescent models. Even though tree-based Markov models of sequence evolution are naturally defined on rooted trees, in many of these works, the tree parameter is assumed to be an unrooted tree. The reason for this is that given an expected site pattern probability distribution from a tree-based Markov model, the location of the root of the tree is not identifiable (see, for example, Section 8.5 in \cite{SempleSteel} or Chapter 15 in \cite{sullivant2018algebraic}). Similarly, with network-based Markov models, even though we define the models on rooted networks, we will only be able to establish generic identifiability when the network parameter is assumed to be a \emph{semi-directed} network. Semi-directed networks are unrooted versions of rooted networks, which retain information about which vertices are reticulation vertices (and which edges are reticulation edges). In \cite{Gross2017DistinguishingPN}, algebraic techniques were used to show that the network parameter is generically identifiable when the underlying Markov process is subject to the Jukes-Cantor (JC) transition matrix constraints and the network parameter is assumed to be a semi-directed network with exactly one cycle of length at least four. Recently, in \cite{hollering2020Identifiability}, this 
result was extended using an algebraic matroid approach to include the Kimura 2-parameter and 
Kimura 3-parameter constraints (K2P, K3P). 

\begin{theorem}
\label{thm: oneRetic}\citep{Gross2017DistinguishingPN,hollering2020Identifiability}
The network parameter of a network-based Markov model under the Jukes-Cantor \citep{Gross2017DistinguishingPN}, Kimura 2-parameter \citep{hollering2020Identifiability}, or Kimura 3-parameter \citep{hollering2020Identifiability} constraints is generically identifiable with respect to the class of models where the network parameter is an $n$-leaf  semi-directed network with exactly one undirected cycle of length of at least four.
\end{theorem}

Still, these identifiability results only apply for networks with a single reticulation vertex.
In this paper, we  prove the following, extending the results to triangle-free, level-1 semi-directed networks, that is, triangle-free semi-directed networks where every undirected cycle 
contains a single reticulation vertex.


\begin{theorem}
\label{thm: main}
The network parameter of a network-based Markov model under the Jukes-Cantor, Kimura 2-parameter, or Kimura 3-parameter constraints is generically identifiable with respect to the class of models where the network parameter is an $n$-leaf triangle-free, level-1 semi-directed network with $r \geq 0$ reticulation vertices.
\end{theorem}

\noindent 
%
To illustrate the implications of Theorem \ref{thm: main}, suppose that $p$ is an expected site pattern probability distribution that belongs to a Markov model on a \emph{rooted} phylogenetic network $N$. If it is known that $N$ is level-1 with triangle-free skeleton and $r$ reticulation vertices, then from $p$, it is possible (almost surely) to determine the unrooted skeleton of $N$ as well as which vertices (edges) are hybrid vertices (edges).


Our proof is largely combinatorial, as we are able to 
use the algebraic results for small networks obtained
in \cite{Gross2017DistinguishingPN} and \cite{hollering2020Identifiability}, in addition to a few new ones, as building blocks.
We begin in Section \ref{sec: preliminaries} by describing more precisely the models we consider as well as the algebraic approach to establishing generic identifiability. 
In Section \ref{sec: Varieties and distinguishability},
we prove a few novel results about the algebra of 4-leaf level-1 networks and collect the other required algebraic results. In Section \ref{sec:combinatorial}, we prove several combinatorial properties of level-1 phylogenetic semi-directed networks that we will need to prove the main result.
Finally, with these results in place, in Section \ref{sec: inductive proof}, we prove Theorem \ref{thm: main}.

\section{Preliminaries} 
\label{sec: preliminaries}

We begin this section by defining the graph theoretic terminology that
we will use throughout the paper.  
Then, in Section \ref{sec: Network based Markov models}, we introduce network-based Markov models on rooted networks, and in Section \ref{sec: semi-directed network models}, we show that we can also define  a network-based Markov model on a semi-directed network.  Finally, we describe the connection between network-based Markov models and algebraic varieties and formally define what it means for two networks to be \emph{distinguishable} and precisely what it means for the network parameter of a class of models to be \emph{generically identifiable}.

\subsection{Graph Theory Terminology}
\label{sec: graph theory terms}

A \emph{(rooted binary) phylogenetic network} ${N}$  on a set of leaves $\mathcal X$ is a rooted acyclic directed graph with no edges in parallel such that the root has out-degree two, each vertex with out-degree zero has in-degree one, the set of vertices with out-degree zero is $\mathcal X$, and all other vertices either have in-degree one and out-degree two, or in-degree two and out-degree one. The \emph{skeleton} of a phylogenetic network is the undirected graph that is obtained from the network by removing edge directions.

A vertex is a \emph{tree vertex} if it has in-degree one and out-degree two.
A vertex is a \emph{reticulation vertex} if it has in-degree two and out-degree one,
and the edges that are directed into a reticulation vertex are called \emph{reticulation edges}. Let $r(\N)$ denote the number of reticulation vertices in network $\N$. Since~$\N$ is binary, it can be shown that it has exactly $2|\mathcal{X}|+2r(\N)-1$ vertices and $|\mathcal{X}|+2r(\N)-1$ internal vertices.
A rooted phylogenetic network with no reticulation vertices is a \emph{rooted phylogenetic tree}.


The \emph{level} of a phylogenetic network is the maximum number of
reticulation vertices in a biconnected component of the network. Of particular interest in this paper are \emph{level-1 networks}, which can also be characterized as phylogenetic networks where no vertex belongs to more than one cycle in the network's skeleton \citep{RosselloValienteGalled}.

More specifically, we will be concerned with a particular kind of level-one network, in which only the reticulation edges are directed. 

\begin{definition}
\label{defn: semidirected}
A \emph{semi-directed network} is a mixed graph  obtained from a phylogenetic network by undirecting all non-reticulation edges, suppressing all vertices of degree two, and identifying parallel edges.
\end{definition}

\noindent Note that deciding whether a mixed graph, a graph with some edges directed and others undirected, is a semi-directed network can be done in quadratic time in the number of edges (Corollary 4 of~\cite{huber2019rooting}).  The \emph{unrooted skeleton} of a phylogenetic network is the skeleton of its associated semi-directed network (including leaf labels).

In a semi-directed network, the reticulation vertices are the vertices of indegree two and the level is defined the same as for a rooted phylogenetic network. A \emph{triangle-free level-1 semi-directed network} is a level-1 semi-directed network where every cycle in the unrooted skeleton has length greater than three. We will also refer to level-1 semi-directed networks with exactly one reticulation vertex as \emph{$k$-cycle networks}, where $k$ is the length of the unique cycle
in the unrooted skeleton.

We finish these preliminaries with one additional bit of graph theory terminology that will be useful throughout. Let $A\cup B$ be a partition of $\leafSet$ with $A,B$ non-empty. 
An edge $e$ in  a network $\N$ \emph{separates} $A$ and $B$ if every path (not necessarily directed) between any $a\in A$ and $b\in B$ contains $e$.
If $e$ separates $A$ and $B$ then we call $e$ a \emph{cut-edge} and we say $\N$ has an \emph{$A-B$ split}.

\subsection{Network based Markov models }
\label{sec: Network based Markov models}


We begin this section by describing a model of DNA sequence evolution
along 
an $n$-leaf rooted binary phylogenetic network. For the description below, we assume that the network belongs to the set of \emph{tree-child networks} \citep{Cardona2007}, which contains the set of level-one networks. In a tree-child network, every internal vertex has at least one child vertex that is either a tree vertex or a leaf.

Let $N'$
be an $n$-leaf phylogenetic network and let $\rho$ be the root of the network. \rev{Let $\mathcal S_4$ be the set of $4 \times 4$ (row) stochastic matrices and let $\Delta^d$ be the $d$th dimensional probability simplex, i.e. $\Delta^d:= \{ p \in \mathbb R^d \ : \ p \geq 0, \ \sum_{i=1}^{d} p_i = 1\} \subseteq \mathbb R^{d}$}.
We associate to each node $v$ of $N$ a random variable $X_v$ with state space $\{A,G,C,T\}$, corresponding to the four DNA bases. The nodes of the network, including the interior nodes, represent taxa, and the random variable $X_v$ is meant to indicate the DNA base at the particular site being modeled in the taxon at $v$.
 Now, let \rev{ $\boldsymbol{\pi}= (\pi_A,\pi_G,\pi_C,\pi_T) \in \Delta^3 \subset \mathbb{R}^4$} be the distribution at the root with $\pi_i = P(X_\rho = i)$, and  associate to each edge $e = uv$ of $N'$ a $4 \times 4$ transition matrix \rev{$M^e \in \mathcal S_4$} where the rows and columns are indexed by the elements of the state space. \rev{With $u$ a parent of $v$, the matrix} $M^e_{i,j}$ is equal to the conditional probability 
$P(X_v = j | X_u = i)$. When $N'$ is a rooted tree, the probability of observing a particular $n$-tuple at the leaves of $N'$ is straightforward to compute. 
Letting $V(N')$ be the vertex set of $N'$, we first consider an assignment of states
to the vertices of $N'$ by $\phi: V(N') \to \{A,G,C,T\}$ where $\phi(v)$ is the state of $X_v$. Then, under the assumption of a tree based Markov model, the probability of observing the assignment $\phi$ can be computed using the distribution at the root and the transition matrices. Specifically, letting $\Sigma(N')$
be the set of edges of $N'$, this probability is equal to 

$$
\pi_{\phi(\rho)}
\prod_{e=uv \in \Sigma(N')} \
M^e_{\phi(u),\phi(v)}.
$$
The probability of observing a particular assignment of states at the leaves can be obtained by marginalization, i.e. summing over all possible assignments of states to the internal nodes. 
In particular, if $\omega \in \{A, G, C, T\}^{|\mathcal X|}$ 
%
%
is an assignment of states to the leaves $\mathcal X$ of $N'$ and $\phi(\mathcal X)$ is the restriction of $\phi$  to the entries corresponding to the leaves of $N'$, the probability of observing $\omega$ is then
$$ \sum_{(\phi \ : \ \phi(\mathcal X) = \omega)} 
\pi_{\phi(\rho)}
\prod_{e=uv \in \Sigma(N')} \
M^e_{\phi(u),\phi(v)}.$$

When the rooted network $N'$ contains at least one cycle in its skeleton, there is no longer a unique path between each leaf and the root, and thus reticulation edge parameters are introduced. In this case, suppose $N'$ has $r$ reticulation vertices $v_1, \ldots, v_r$. 
Since each $v_i$ has in-degree two, 
there are two edges, $e^0_i$ and $e^1_i$, 
directed into $v_i$. Assign a parameter $\delta_i \in (0,1)$ to $e^1_i$ and the value $1-\delta_i$ to $e^0_i$. For $1 \leq i \leq r$, 
independently delete $e^0_i$, keeping $e^1_i$, with probability
$\delta_i$, otherwise, delete $e^1_i$ and keep $e^0_i$.
Intuitively, the parameter $\delta_i$ corresponds to the probability that a particular site was inherited along edge $e^1_i$. 
Encode this set of choices with a binary vector $\sigma \in \{0,1\}^r$ where a $0$ in the $i$th coordinate indicates that edge $e^0_i$ was deleted. Since $N'$ is assumed to be a tree-child network, after deleting the $r$ edges, the result is a rooted $n$-leaf tree $T_{\sigma}$. Since there are four DNA bases and $n$ leaves of the network, there are $4^n$ possible \emph{site-patterns}, or assignment of states, that could be observed
at the leaves of $N'$. The probability of observing the site-pattern $\omega$ is
%
%
 \begin{equation} \label{eq:probs}
p_{\omega}= \sum_{\sigma \in \{0,1\}^r}
\left ( \prod_{i=1}^r \delta_i ^{1-\sigma_i}(1-\delta_i)^{\sigma_i}\right ) 
\sum_{(\phi \ : \ \phi(\mathcal X) = \omega)} \pi_{\phi(\rho)} 
\prod_{e=uv \in \Sigma(T_{\sigma})} \
M^e_{\phi(u),\phi(v)} .
\end{equation}






\begin{figure}
  \begin{center}
    \includegraphics[width=\textwidth]{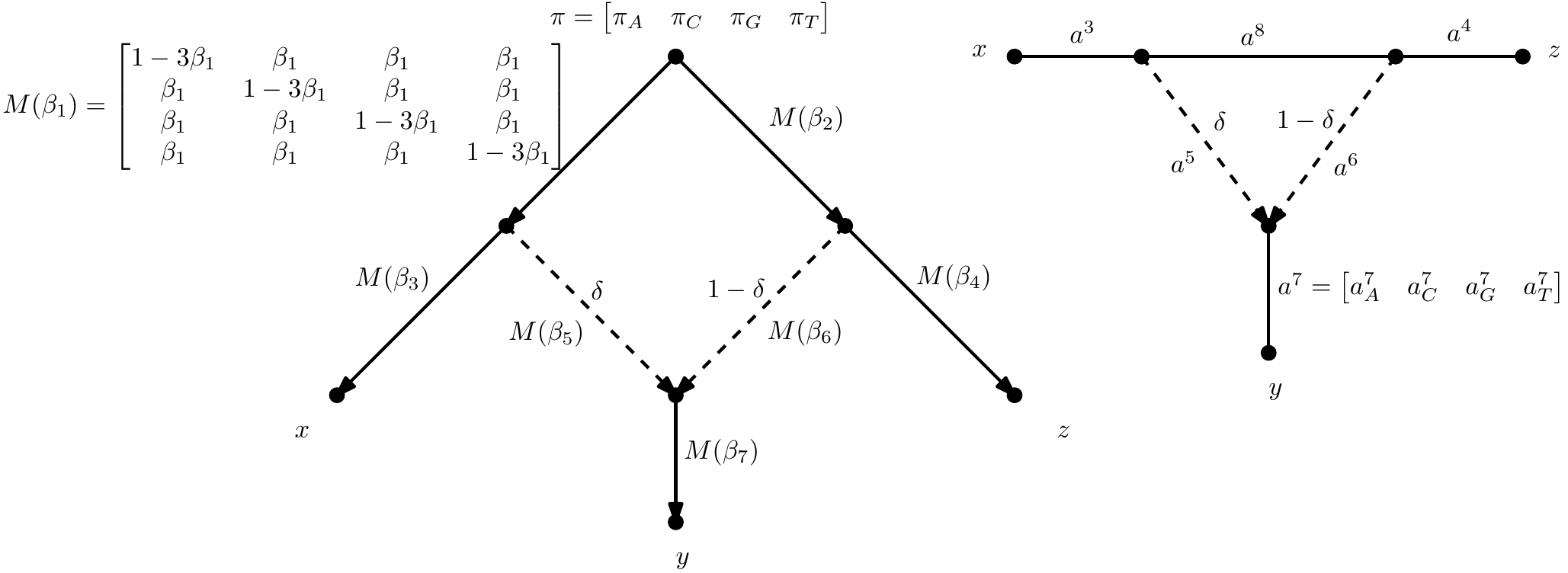}
  \end{center}
  \caption{\label{fig:parameters} On the left is an example of a phylogenetic network with stochastic transition matrices assigned to each edge and reticulation parameters assigned to the two reticulation edges; we denote the edge transition matrices using $M(\beta_i)$ rather than $M^{e_i}$ to indicate the dependence on the parameter $\beta_i$. The transition matrices all satisfy the Jukes--Cantor constraints. On the right is the semi-directed network obtained by unrooting the network on the left.  
  Each edge of the semi-directed network is labeled by a vector of Fourier parameters.
  Reticulation edges are represented by dashed edges.}
\end{figure}

While seemingly complicated, the above expression is a polynomial in the numerical parameters of the model: the root distribution, the entries of the transition matrices, and the reticulation edge parameters. 
 Thus the map defined by the network $N'$
$$\psi_{N'}: \theta_{N'} \to \Delta^{4^n -1},$$
from the numerical parameter space \rev{$ \theta_{N'}:= \Delta^3 \times (\mathcal S_4)^{|\Sigma(N')|} \times (0,1)^r$} to the probability
simplex  $\Delta^{4^n -1}$ is a polynomial map.  
The image of the map $\psi_{N'}$ is called
the \emph{model associated to $N'$}, denoted 
$\mathcal{M}_{N'}$. Note the model $\mathcal{M}_{N'}$ is the set of all possible probability distributions obtained by fixing the network $N'$ and varying the numerical parameters. See Figure~\ref{fig:parameters} for an example of a network with its numerical parameters.


When we place no restrictions on the entries of the transition matrices (other than that they are stochastic) the underlying substitution process is known as the \emph{general Markov model}. 
Network-based phylogenetic models with a general Markov substitution process are studied for example in \cite{casanellas2020rank}. 
However it is quite common in phylogenetics to consider models with 
additional constraints, effectively reducing the dimension of the parameter space $\theta_{N'}$.  For example, in the Kimura 3-parameter DNA substitution model, the 
root distribution is uniform and each transition matrix is assumed to have the 
following form, where the rows and columns are indexed by the DNA bases $A, G, C, T$,
$${\footnotesize \begin{pmatrix}
     \alpha & \beta  & \gamma & \delta \\
     \beta  & \alpha & \delta & \gamma \\
     \gamma & \delta & \alpha & \beta  \\
     \delta & \gamma & \beta  & \alpha
    \end{pmatrix}}. $$
%
In the Kimura 2-parameter model (K2P), and Jukes-Cantor models, additional restrictions are placed on the entries of the transition matrices ($\gamma = \delta$ for K2P and $\beta = \gamma = \delta$ for JC).
%
%
    %
    In order to not overload the word ``model,"  we will refer to these restrictions on the transition matrices as \emph{constraints}. For example, we will refer to the image of $\psi_{N'}$ under the Jukes-Cantor DNA substitution model as the \emph{model associated to $N'$ under the Jukes-Cantor constraints.} 
    
    We end this section on network-based Markov models by noting that there exist other natural extensions of tree-based Markov models. For example, in \cite{francis2018identifiability}, the authors consider a network model adapted from \cite{thatte2013reconstructing} and are able to establish identifiability for the entire class of tree-child networks. The stronger identifiability results come at the expense of some modeling flexibility, but the difference can illustrate the possible gains that can be made by considering different processes. 


\subsection{Semi-directed Network Models}
\label{sec: semi-directed network models}

In this section, we show how to associate a model $\mathcal{M}_N$ to a phylogenetic semi-directed network $N$ for the group-based models considered in this paper. We will see that for a given set of constraints (JC, K2P, K3P), if $N'$ is a phylogenetic network and $N$ is the semi-directed network attained from $N'$ as in Definition \ref{defn: semidirected}, then $\mathcal{M}_N = \mathcal{M}_{N'}$. 
We start by showing that the model associated to a rooted network $N$ does not depend on the location of the root. 
Then, we show that the associated model does not change if we suppress degree two vertices or remove parallel edges in the network. Thus, the phylogenetic semi-directed network $N$ contains all of the information necessary to recover $\mathcal{M}_{N'}$.

For a tree-based phylogenetic model under the Jukes-Cantor, Kimura 2-parameter, or Kimura 3-parameter constraints, we may relocate the root and suppress vertices of degree two without changing the underlying model (see, for example, Section 8.5 in \cite{SempleSteel} or Chapter 15 in \cite{sullivant2018algebraic}).
That we can relocate the root is easily observed since each of the transition matrices is symmetric and the root distribution is uniform, so that $\pi_iM_{i,j} = \pi_jM_{j,i}.$ 
To see that we may suppress vertices of degree two without changing the model, suppose the edges $e$ and $f$ are incident to a vertex of degree two and that the Markov transition matrices $M^e$ and $M^f$ satisfy the Jukes-Cantor, Kimura 2-parameter, or Kimura 3-parameter constraints. Then the transition matrix $M^eM^f$ will satisfy the same constraints, so we may suppress the vertex of degree two and assign this transition matrix to the newly created edge to obtain the same site pattern probability distribution from the model. These results imply that the location of the root of the rooted tree parameter in a tree-based Markov model cannot be identified from an expected site-pattern in the model. Or, viewed another way, these results mean that we can associate a tree-based Markov model to an unrooted tree and consider the tree parameter in a tree-based Markov model to be an unrooted tree.

A similar result holds for the network-based Markov models considered in this paper. For a fixed choice of parameters in a network model, the associated site pattern probability distribution is the weighted sum of site-pattern probability distributions from the constituent tree models. The weights are determined by the reticulation edge parameters. Since relocating the root in each of the trees does not affect the tree models, the network model will remain the same if we relocate the root of the network and redirect the edges in any way that preserves the direction of the reticulation edges. For example, in the rooted network in Figure \ref{fig:parameters}, we could suppress the existing root vertex, subdivide the edge directed into the leaf vertex labeled by $z$ to create a new root, and then redirect edges away from the new root in a way that preserves the directions of the reticulation edges. 
%

If a child of the root vertex is a reticulation vertex, then unrooting and suppressing the root will may result in a pair of parallel reticulation edges in the semi-directed network. However, under the JC, K2P, and K3P constraints, we may identify any pair of parallel edges without altering the model. The reason for this is that the sets of transition matrices under each of these constraints are closed under convex sums. So if a network contains a set of parallel reticulation edges with transition matrices $M^e$ and $M^f$, we can replace these edges with a single edge with transition matrix $\delta M^e + (1 - \delta)M^f$ and obtain the same site-pattern probability distribution in the model, where~$\delta$ is the reticulation edge parameter for the edge~$e$.  

Together, these arguments give us the following proposition.

\begin{prop} 
\label{prop:semidirectednetworks}
Let $N_1'$ and $N_2'$ be two tree-child phylogenetic networks with
associated phylogenetic semi-directed networks $N_1$ and $N_2$. Under the JC, K2P, or K3P constraints, if $N_1 = N_2$ then $\mathcal{M}_{N_1'} = \mathcal{M}_{N_2'}$.
\end{prop}
%
%
Thus, the model associated to a rooted phylogenetic network is entirely determined by the associated phylogenetic semi-directed network. Although we note that the arguments above are specific to the JC, K2P, and K3P constraints, and similar arguments might not work for other network-based Markov models.
%
%

Proposition \ref{prop:semidirectednetworks} suggests that we may regard the network parameter of a network-based Markov model as a phylogenetic semi-directed network. Given a phylogenetic semi-directed network $N$, we can determine the model $\mathcal{M}_{N}$ by choosing any rooted network $N'$ for which $N$ is the associated semi-directed network and defining $\mathcal{M}_{N} := \mathcal{M}_{N'}$. 
%
%
Therefore, for the rest of this paper, we will assume that the network
parameter of each model is an $n$-leaf phylogenetic semi-directed network. 
Indeed, this is necessary to obtain any identifiability results, as the location of the root in a rooted network is not identifiable from an expected site pattern probability distribution in the model.



\subsection{Markov Models as Algebraic Varieties}
\label{sec: Markov Models as Algebraic Varieties}

In this paper, we prove generic identifiability using tools from combinatorics and computational algebraic geometry. 
In order to understand ${\mathcal M}_N = Im(\psi_N)$ within an algebraic-geometric framework, we consider the complex extension of $\psi_N$, which we denote as $\psi'_N$. 

Let 
$\mathbb C[p_{\omega}: \omega \in \{A,G,C,T\}^n]$ 
%
%
be the set of all polynomials on $4^n$ variables with coefficients in $\mathbb C$.  The \emph{ideal associated to $\mathcal M_N$} is the set of polynomials that vanish on the image of $\psi'_N$, i.e.

$$ \mathcal I_{N} := \{ f \in \mathbb{C}[p_{\omega}: \omega \in \{A,G,C,T\}^n] \ : \ f({p}) = 0 \ \forall {p} \in Im(\psi_N')\}.$$

The elements of $\mathcal I_{N}$ are called \emph{phylogenetic invariants}.  Each polynomial in $\mathcal I_N$ vanishes on ${\mathcal M}_N$, that is, each polynomial yields
zero when we substitute the entries of any probability distribution ${p} \in {\mathcal M}_N$. Phylogenetic invariants are the defining polynomials of the variety $\mathcal V_{N}$ associated to $\mathcal M_N$, which we will refer to as the \emph{network variety}. Specifically, 
$$
\mathcal V_{N} 
:= \mathcal V(\mathcal I_{N})=
\{ p \in \mathbb C^{4^n} \ : \ f(p) =
0 \text{ for all } f \in \mathcal I_{N} \}.$$


Elements of $\mathcal I_N$ are polynomial relationships among the entries of $p$ that hold for all distributions $p \in \mathcal M_N$.  If we look back at equation \eqref{eq:probs}, it is reasonable to assume that such relationships may be quite complicated since each probability coordinate $p_{\omega}$ is parameterized by a polynomial that is the sum of $2^r4^{(n + 2r -1)}$ terms. 
Because of this, we perform a linear change of coordinates on both the parameter space
and the image space called the Fourier-Hadamard transform \citep{Evans1993, hendy1996Complete}. 
After the transform, the invariants
are expressed in the ring of \emph{$q$-coordinates}, 
$$\mathbb{C}[q_{\omega}: \omega \in \{A,G,C,T\}^n].$$
As an example of how the Fourier-Hadamard simplifies the resulting algebra,
for a tree-based phylogenetic model, the parameterization of each 
$q$-coordinate is a monomial in the \emph{Fourier parameters} 
and the phylogenetic tree ideal is generated by binomials.
Working in the transformed coordinates is common \rev{when working with group-based models} and it
is what enables us to compute the required network invariants. 
While the details of the Fourier-Hadamard transform are outside the scope of this paper,
we give here a brief description of how to parametrize a phylogenetic network model
under the Jukes-Cantor, Kimura 2-parameter, and Kimura 3-parameter constraints.
More details can be found in \cite{Toric2005} and Chapter 15 of \cite{sullivant2018algebraic}.  

First, we will describe how to determine the Fourier parametrization
of a phylogenetic tree, $T$. As in \cite{Toric2005} and \cite{sullivant2018algebraic}, we begin by identifying the four DNA bases with elements of the group
$\mathbb{Z}_2 \times \mathbb{Z}_2$ as follows $A = (0,0)$, $G = (1,0)$, $C = (0,1)$ and $T = (1,1)$.
Under the Kimura 3-parameter constraints, there are then four Fourier parameters 
associated to each edge $i$, denoted as $a^i_A$, $a^i_G$, $a^i_C$, and $a^i_T$ (after transformation, the stochastic condition on the transition matrices forces $a^i_A = 1$).
Letting $\omega$ be the site pattern $(g_1, g_2, \ldots, g_n)$, the parametrization is then given by

 \begin{displaymath}
   q_{\omega}  = \left\{
     \begin{array}{cl}
       \displaystyle \prod_{e \in \Sigma(T)} a^{e}_{\sum_{j \in Y} g_j} & \text{     if }\displaystyle \sum_{j = 1}^{n} g_j= 0 \\
       & \\
       0 & \text{     otherwise.}
     \end{array}
   \right.
\end{displaymath} where  $\Sigma(T)$ is the set of edges of $T$ and 
$Y - Z$ is the split induced by $e$ in $T$.
All addition is in the group $\mathbb{Z}_2 \times \mathbb{Z}_2$.

Notice that this is a monomial, in which there is one parameter associated to each edge of the
tree $T$. In order to parametrize a phylogenetic network, we take the sum of the monomials
corresponding to all $2^r$ trees created by removing reticulation edges from the network. 
The monomials are weighted by the corresponding reticulation edge parameters. 

\subsection{Generic identifiability} 

A model-based approach to network inference selects the model
from a set of candidate models that best fits the observed data according to some criteria and returns the network parameter of this model. In our setting, the observed data are the aligned DNA sequences of the taxa under consideration, from
which we construct the observed site pattern probability distribution. 
In the ideal setting, if we had access to infinite noiseless data generated by a network-based Markov model, then the observed site pattern distribution would be equal to an expected site pattern distribution in the model. Inferring the correct network parameter in this case would be as simple as determining which model from a set of candidate models the site pattern probability distribution belongs to.
However, even in this idealized setting, it may be that the observed site pattern distribution belongs to the models corresponding to two distinct networks, making it impossible to determine which network produced the data. Thus, a desirable theoretical property for a class of network models is that each distribution in one of the models belongs to no other model, or that the network parameter be \emph{identifiable}.




Let $\mathcal{N}$ be a set of leaf-labelled networks. More formally, the condition that the network parameter is identifiable with respect to a collection of models $\{\mathcal M_N\}_{N\in \mathcal N}$ is equivalent to the condition that for all distinct
$N_1, N_2 \in \mathcal{N}$, $\mathcal{M}_{N_1} \cap \mathcal{M}_{N_2} = \emptyset$, meaning the two models 
do not intersect.
Since this notion of identifiability is rather strong, the more practical notion of  \emph{generic identifiability} is more commonly explored.

\begin{definition}
\label{defn: genericallyidentifiable}
Let $\{\mathcal M_N\}_{N\in \mathcal N}$ be a class of phylogenetic network models. The network parameter  is \emph{generically identifiable with respect to the class} 
$\{\mathcal M_N\}_{N\in \mathcal N}$ if given any two distinct
$n$-leaf networks
$N_1, N_2 \in \mathcal N$, 
the set of numerical parameters in $\theta_{N_1}$ that 
$\psi_{{N}_1}$ maps into $\mathcal{M}_{{N}_2}$ is a set of Lebesgue measure zero.
\end{definition}

To establish generic identifiability, 
we can use algebraic geometry by
considering the family of irreducible algebraic varieties
$\{\mathcal V_N\}_{N \in \mathcal{N}}$,
where $\mathcal V_N$ is the network variety associated to $N$.
Generic identifiability is then closely related to 
the concept of \emph{distinguishability}.

\begin{definition}\citep{Gross2017DistinguishingPN}
 Two distinct $n$-leaf networks $N_1$ and $N_2$ are \emph{distinguishable} if $\V_{N_1} \cap \V_{N_2}$  is a proper subvariety of $\V_{N_1}$ and $\V_{N_2}$,
 that is, $\V_{N_1} \not\subseteq \V_{N_2}$ and $\V_{N_1} \not\supseteq \V_{N_2}$.
 Otherwise, they are \emph{indistinguishable}. 
\end{definition}

\begin{prop}
\cite[Proposition 3.3]{Gross2017DistinguishingPN}
\label{prop: identifiability}
Let $\{\mathcal M_N\}_{N\in \mathcal N}$ be a class of phylogenetic network models. The network parameter of a phylogenetic network model is generically identifiable with respect to $\{\mathcal M_N\}_{N\in \mathcal N}$ if given any two distinct $n$-leaf networks  $N_1, N_2 \in \mathcal N$, the networks $N_1$ and $N_2$  are distinguishable.
\end{prop}

The condition that the network parameter be generically identifiable
then amounts to showing that for all
$N_1, N_2 \in \mathcal{N},$ the networks $N_1$ and $N_2$ are distinguishable, or equivalently, $\V_{N_1} \not\subseteq \V_{N_2}$ and $\V_{N_1} \not\supseteq \V_{N_2}$.
Proving that this condition is satisfied can then be done either by explicit computation of 
the ideals associated to $N_1$ and $N_2$ (as in \cite{Gross2017DistinguishingPN}), or by arguing
that certain phylogenetic invariants must exist (as in \cite{hollering2020Identifiability}).

\section{Distinguishability of 4-leaf semi-directed networks}
\label{sec: Varieties and distinguishability}

Our aim is to prove Theorem \ref{thm: main}, by showing that any two distinct $n$-leaf $r$-reticulation triangle-free level-$1$ semi-directed networks are distinguishable.
In order to show this, we will require a number of results concerning $4$-leaf networks which we prove in Lemma \ref{lem:buildingBlocks} below.

Up to leaf relabeling, there are six different 4-leaf level-$1$ semi-directed networks which are depicted
in Figure \ref{fig:UnrootedShapes}. 
In Lemma 1, we assume that $N_1$ and $N_2$ are two distinct 4-leaf semi-directed networks. We then consider all cases where $N_1$ and $N_2$ are each either a quartet tree ($Q$), a single triangle network ($\Delta$), a double-triangle network ($DT$), or a 4-cycle network ($4C$), and compare the resulting varieties.
We only need to consider four possibilities for each of $N_1$ and $N_2$, because under the JC, K2P, and K3P constraints, the variety of a triangle or double-triangle semi-directed network is determined by the unrooted skeleton of the network.
This can be shown by first observing that 
under the JC, K2P, and K3P models, the 
ideals of all of the 3-leaf semi-directed triangle networks are identical.
The proof then follows by applying the same
\emph{toric fiber product} argument that is described in the remark following
Proposition 4.5 in \cite{Gross2017DistinguishingPN}.

The results of Lemma \ref{lem:buildingBlocks} are summarized in Table \ref{tab:InclusionOverview} and the caption of that table contains the key to the symbols. To give a couple of examples, part (ii) of the lemma corresponds to the $(2,2)$ entry of the table. The $\sim$ symbol indicates that the networks are distinguishable, but only if $N_1$ and $N_2$ have distinct unrooted skeletons. The results of part (iii) of the lemma are represented by the entries $(4,1)$ and $(4,2)$ (when $k_1= 4$ and $N_1$ is a 4-cycle network) and by $(2,1)$ (when $k_1 = 3$ and $N_1$ is a 3-cycle, or triangle network). And of course, these results are also represented by the entries $(1,4)$, $(2,4)$, and $(2,1)$ when the roles of $N_1$ and $N_2$ are reversed.

\begin{figure}
    \centering
    \begin{subfigure}[b]{0.2\textwidth}
        \includegraphics[width=0.8\textwidth]{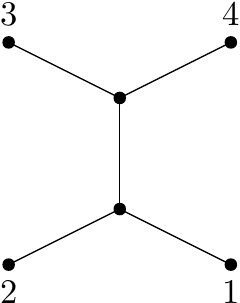}
        \caption{The quartet tree (Q)}
        \label{fig:Unrooted4Tre}
    \end{subfigure}
    ~ 
    \begin{subfigure}[b]{0.2\textwidth}
        \includegraphics[width=0.8\textwidth]{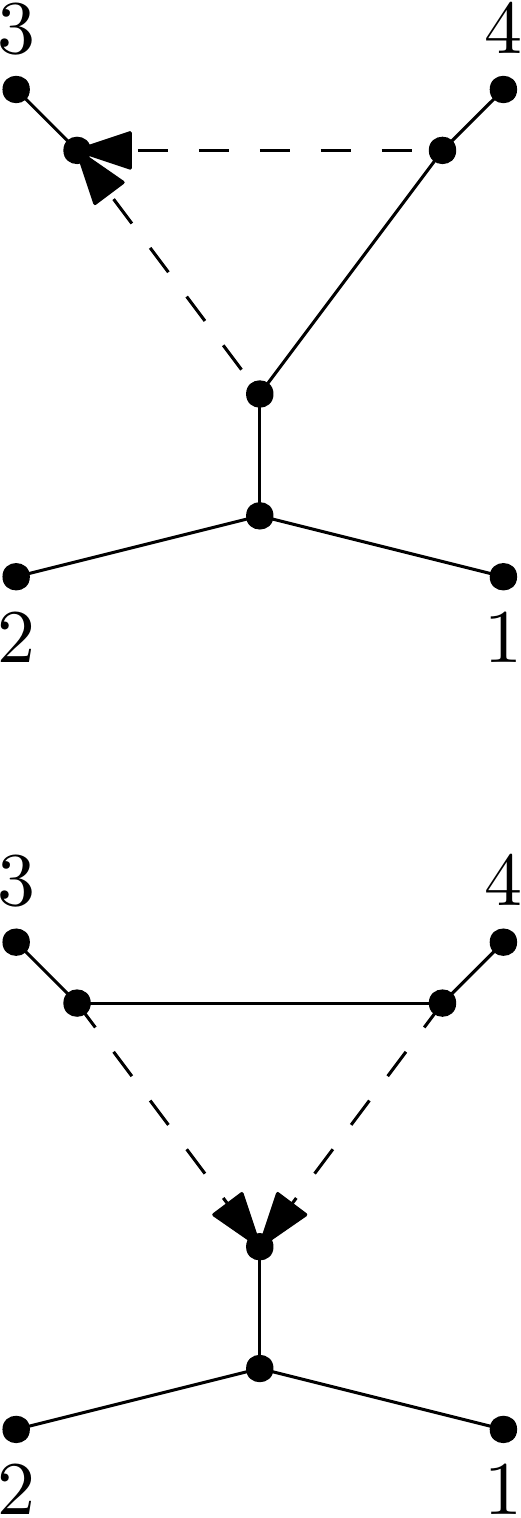}
        \caption{The single-triangle ($\Delta$)}
        \label{fig:UnrootedTriangleGraph}
    \end{subfigure}
    ~
    \begin{subfigure}[b]{0.2\textwidth}
        \includegraphics[width=0.8\textwidth]{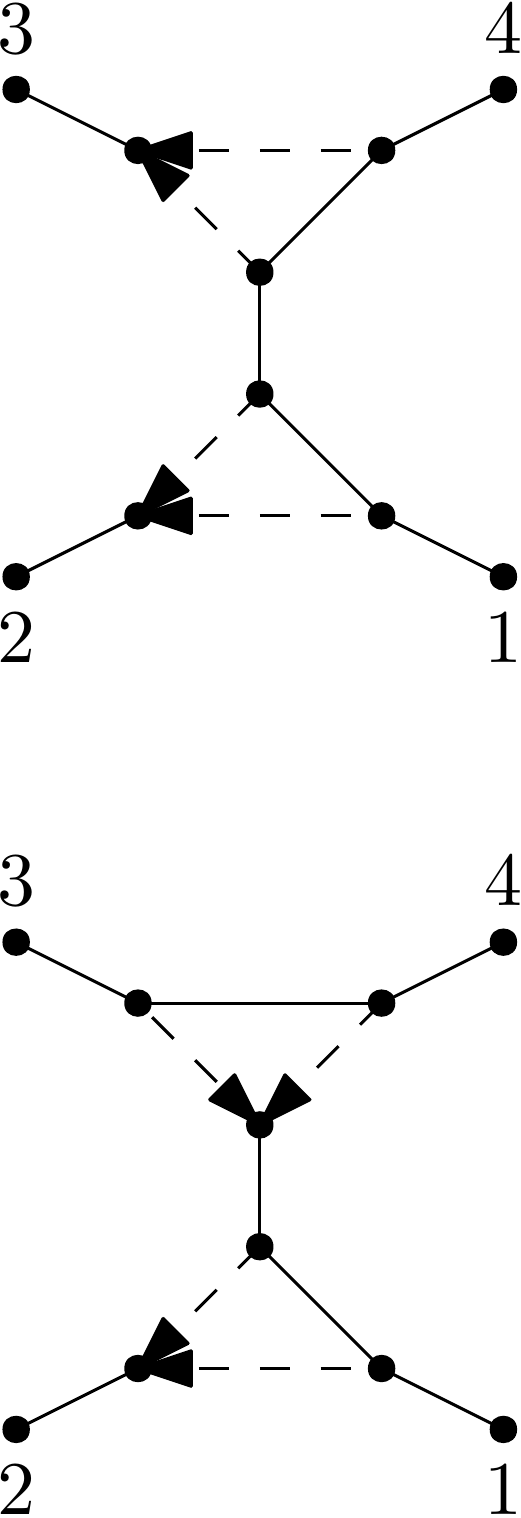}
        \caption{The double-triangle (DT)}
        \label{fig:UnrootedDoubleTriangle}
    \end{subfigure}
    ~
    \begin{subfigure}[b]{0.2\textwidth}
        \includegraphics[width=0.8\textwidth]{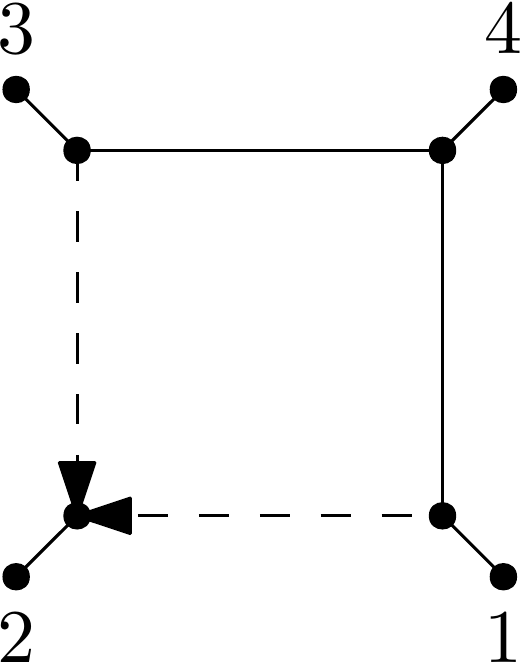}
        \caption{The $4$-cycle ($4C$)}
        \label{fig:Unrooted4Cycle}
    \end{subfigure}
    \caption{
    All possible semi-directed level-1 networks on four leaves (up to relabeling of leaves), grouped by their unrooted skeletons.
    }\label{fig:UnrootedShapes}
\end{figure}

\begin{lemma}\label{lem:buildingBlocks}
Let $N_1$ and $N_2$ be distinct $4$-leaf level-1 semi-directed networks.
Then under the JC, K2P, or K3P constraints:
\begin{enumerate}[(i)]
\item If $N_1$ and $N_2$ are both trees, 
then $N_1$ and $N_2$ are distinguishable;

\item If $N_1$ and $N_2$ are both single-triangle networks
and have different (leaf-labelled) unrooted skeletons, 
then $N_1$ and $N_2$ are distinguishable;

\item If $N_1$ is a $k_1$-cycle network with 
$k_1 \leq 4$ and $\N_2$ is a tree or a $k_2$-cycle network with $k_2 < k_1$,
then $\V_{N_1} \not\subseteq \V_{N_2}$;

\item If $N_1$ and $N_2$ are both $4$-cycle networks, then $N_1$ and $N_2$ are distinguishable;

\item If $N_1$ is a double-triangle network and $N_2$ a single-triangle network or a tree, then $\V_{N_1} \not\subseteq \V_{N_2}$;
\item If $N_1$ is a double-triangle network and $N_2$ is a $4$-cycle network, 
then $N_1$ and $N_2$ are distinguishable;

\item If $N_1$ and $N_2$ are both double-triangle networks 
and have different (leaf-labelled) unrooted skeletons, 
then $N_1$ and $N_2$ are distinguishable. 

\end{enumerate}
See Table~\ref{tab:InclusionOverview} for an overview.
\end{lemma}

\begin{table}
 \caption{An overview of  Lemma~\ref{lem:buildingBlocks} results for two distinct 4-leaf level-1 semi-directed networks~$N_1$ and~$N_2$. The two networks are represented by the row for~$N_1$ and the column for~$N_2$, and each element in the~$4\times4$ grid indicates whether the two networks are distinguishable ($\surd$), the variety of one network is not contained in that of the other ($\not\subseteq$ means $\V_{N_1} \not\subseteq \V_{N_2}$, and $\not\supseteq$ means $\V_{N_1} \not\supseteq \V_{N_2}$), or the two networks are distinguishable if the unrooted skeletons are different ($\sim$).}
 \label{tab:InclusionOverview}
 \begin{center}
 \begin{tabular}{| lc|c|c|c|c| } 
  \cline{3-6}
  \multicolumn{2}{c|}{} & \multicolumn{4}{ c| }{$N_2$}\\
  \multicolumn{2}{c|}{} & \multicolumn{1}{c}{Q} & \multicolumn{1}{c}{$\Delta$} & \multicolumn{1}{c}{$DT$} & $4C$ \\
  \hline 
  &Q  & $\surd$ & $\not\supseteq$  & $\not\supseteq$ & $\not\supseteq$ \\
  \cline{3-6}
  $N_1$ &$\Delta$  & $\not\subseteq$  & $\sim$& $\not\supseteq$ & $\not\supseteq$ \\
  \cline{3-6}
  & DT  & $\not\subseteq$  & $\not\subseteq$ & $\sim$ &  $\surd$ \\
  \cline{3-6}
  &4C  & $\not\subseteq$  & $\not\subseteq$ &  $\surd$ &  $\surd$ \\
  \hline
 \end{tabular}
 \end{center}
\end{table}

The proof of Lemma \ref{lem:buildingBlocks} will be given below. 
We first outline the proof strategy. Some parts of the lemma
will follow immediately from results in \cite{Gross2017DistinguishingPN} and \cite{hollering2020Identifiability}. 
In \cite{Gross2017DistinguishingPN}, the proofs were obtained by computing Gr\"{o}bner bases for
all of the ideals involved and then comparing
the ideals. However, this was only
possible because the constraints considered were the 
Jukes-Cantor constraints, the most restrictive that we consider.
In \cite{hollering2020Identifiability}, the
authors extend the results to the 
K2P and K3P constraints using a method based on the theory
of algebraic matroids. This method is preferable
when there are fewer constraints since the Gr\"{o}bner bases computations 
are difficult if not impossible to carry out.
Here, we find the required
invariants by modifying this method slightly.
Specifically, we apply the random search strategy
described in that paper to locate small subsets 
of variables that are likely to contain distinguishing
invariants. We then perform our computations in 
a much smaller subring of the original variables.
This greatly reduces the size of the required computations and allows us to generate specific invariants without computing Gr\"{o}bner bases for the ideals.

In order to reduce the total number of invariants 
required to prove each part, we take advantage
of the symmetry between networks.
As an example, suppose that the statement in part (vii) is
false. Then there must exist two double-triangle 
networks with distinct skeletons,
$N_1$ and $N_2$, 
that are not distinguishable. 
All of the network varieties are
parameterized, and hence irreducible as algebraic varieties, 
which means we may assume that if 
two networks are not distinguishable 
then one is contained in the other. Thus,
without loss of generality, we may assume that 
$\V_{N_1} \subseteq \V_{N_2}$, which implies
the reverse inclusion of ideals, $\I_{N_2} \subseteq \I_{N_1}$. 
Up to relabeling, every double-triangle network 
has the same unrooted skeleton. Thus, 
we can obtain any arbitrary double-triangle network $\hat N_2$ from  $N_2$ by permuting leaf labels. If we apply the
same permutation to the leaf labels of $N_1$, we
obtain another double-triangle network $\hat N_1$
for which $\I_{\hat N_2} \subseteq \I_{\hat N_1}$.
Since our choice of $\hat N_2$ is arbitrary, if we can
show that there is a single
double-triangle network with ideal not contained 
in the ideal of any other double-triangle
network, then we arrive at a contradiction, and have
thus proven part (vii). Therefore, in order to prove part (vii), it will suffice to produce a single invariant that vanishes on exactly one of the double-triangle network varieties. 
A similar argument applies in each of the other parts. 

In order to prove some parts of the lemma, we require two or more invariants to distinguish all of the relevant networks,
though all parts can be proven using some combination of just the following six polynomial invariants:
\begin{align*}
g_1 = &q_{ATTA}q_{CCGG}q_{GATC} - q_{AAGG}q_{CTTC}q_{GCTA},\\
g_2 = &q_{CTTC} - q_{GCGC}, \\
g_3 = &q_{CAGT}q_{GTCA}q_{TGAC} - q_{CACA}q_{GTGT}q_{TGAC} - q_{CAGT}q_{GTAC}q_{TGCA} + \\               
&q_{CAAC}q_{GTGT}q_{TGCA} + q_{CACA}q_{GTAC}q_{TGGT} - q_{CAAC}q_{GTCA}q_{TGGT}, \\
g_4 = &q_{AACC}q_{CGCG}q_{GAGA}q_{TAAT} - q_{AACC}q_{CGAT}q_{GAGA}q_{TACG} + \\
      &q_{AACC}q_{CAGT}q_{GGAA}q_{TACG} - q_{AAAA}q_{CAGT}q_{GGCC}q_{TACG}, \\
g_5 = &q_{AAAA}q_{GACT}q_{GCGC} - q_{AAGG}q_{TAAT}q_{TGCA}, \\
g_6 = &q_{AAGG}q_{GATC}q_{TAAT} - q_{AATT}q_{GAAG}q_{TAGC}. \\
\end{align*}
In the supplementary Macaulay2 \citep{M2} files,
available at \\ 
\begin{center} 
    \href{http://github.com/colbyelong/DistinguishingLevel1PhylogeneticNetworks}{github.com/colbyelong/DistinguishingLevel1PhylogeneticNetworks},
\end{center}
we provide the code to verify
that these polynomials vanish or do not
vanish on the referenced varieties as claimed.

\begin{proof}[Proof of Lemma \ref{lem:buildingBlocks}]
Statement (i) is a well-known result for the JC, K2P, and K3P constraints and can be verified using the \emph{Small trees catalog} \citep{smalltreescatalog}.
For the Jukes-Cantor constraints, (ii)-(iv) follow from Proposition 4.6,
Corollary 4.8, and Corollary 4.9
in \cite{Gross2017DistinguishingPN}. 

To prove (ii) for the K2P and K3P constraints
we require a set of invariants that vanishes on 
exactly one of the single-triangle networks.
The set $\{g_1\}$ is confirmed to be such a set 
for both constraints in the supplementary files.  
Statements (iii) and (iv) are proven for the
K2P and K3P models by Lemmas 28 and 29 in
\cite{hollering2020Identifiability}.

To prove (v), we require a set of invariants that
vanishes on one of the tree varieties, but on
none of the double-triangle network varieties,
and a set of invariants that vanishes on 
one of the single-triangle networks varieties,
but on none of the double-triangle network varieties.
The set $\{g_1\}$ is shown to be the required set
for both parts under K2P and K3P, and the set
$\{g_1, g_2\}$ works for the JC constraints.

For (vi), we must first show that there is a set of invariants
that vanishes on one of the double-triangle
network varieties but on none of the 4-cycle network
varieties. The set $\{g_3\}$ works for all constraints
and thus establishes that if $N_1$
is a double-triangle and $N_2$ is a 4-cycle
network, then $V_{N_2} \not \subseteq \V_{N_1}$.
We prove that $V_{N_1} \not \subseteq V_{N_2}$,
and hence that the networks are distinguishable,
by constructing a set of invariants that vanishes
on one of the 4-cycle network varieties but
on none of the double-triangle network varieties.
For the JC constraints, this set is $\{g_4, g_5\}$.
For the K2P and K3P constraints, this set is
$\{g_4, g_6\}$.

The invariant $g_3$ also establishes (vii), 
since it vanishes on exactly one of the double-triangle networks under JC, K2P, and
K3P.
\qed
\end{proof}

We also need a result on $4$-leaf networks that does not fit into Table~\ref{tab:InclusionOverview}.  To state this result we first need some definitions concerning the type of splits
in a network.


\begin{definition}
For networks $\N_1$ and $\N_2$, we say  $X-Y$ is a \emph{common split} if $X-Y$ is a split in both $\N_1$ and $\N_2$; it is \emph{non-trivial} if $|X|, |Y| \geq 2$.
 Two splits $X-Y$ in $\N_1$ and $A-B$ in $\N_2$ are \emph{conflicting} if $X \cap A, X \cap B,  Y \cap A,  Y\cap B$ are all non-empty.
\end{definition}

\begin{lemma}\label{lem:conflictingSplits}
 Let $\N_1$ and $\N_2$ be distinct $4$-leaf level-1 semi-directed networks. If $\N_1, \N_2$ have conflicting splits, then $\N_1$ and $\N_2$ are distinguishable under the JC, K2P, or K3P constraints.
\end{lemma}

%

\begin{proof}
 Note that 4-cycles have no non-trivial splits, so we just need to compare trees, single-triangle networks, and double-triangle networks. Moreover, Table
 \ref{tab:InclusionOverview} shows that we only need to verify that 
 $\mathcal{V}_{N_1} \not \subseteq \mathcal{V}_{N_2}$ 
 in the following cases: 
 \begin{enumerate}[(i)]
     \item when $N_1$ is a tree or triangle network and $N_2$ is a double-triangle network with a conflicting split and
     \item when $N_1$ is a tree and $N_2$ is a triangle network with a conflicting split.
 \end{enumerate}
 The invariant  $g_3$ can be used to verify case (i) for all three constraints.
 The invariant $g_2$ can be used to verify case (ii) for JC, and $g_1$ can be used to verify case (ii) for K2P and K3P.
 \qed
\end{proof}

Finally we require 
Lemma~\ref{lem:subnetVarieties},
 which allows us to use the above small networks as building blocks to prove the claim about larger networks. 
 To state Lemma~\ref{lem:subnetVarieties}, we first define the \emph{restriction} of a network to a subset of leaves.

\begin{definition}
\label{defn: Restriction}
Let $N$ be an~$n$-leaf semi-directed network on~$\leafSet$
and let $A \subseteq \leafSet$. The \emph{restriction of $N$
to $A$} is the semi-directed network $N|_{A}$
obtained by taking the union of all directed paths between leaves
in $A$
(where undirected edges are treated as bidirected) and
then suppressing all degree two vertices and removing parallel edges.
\end{definition}



Lemma \ref{lem:subnetVarieties} is essentially a one-way version of Proposition 4.3 from \cite{Gross2017DistinguishingPN}, and we use a piece of the proof of that proposition below.
 
 \begin{lemma}\label{lem:subnetVarieties}
  Let $\N_1$ and $\N_2$ be distinct~$n$-leaf semi-directed networks on~$\leafSet$ and let $A \subseteq \leafSet$. If $\V_{\N_1|_A} \not\subseteq \V_{\N_2|_A}$, then $\V_{\N_1} \not\subseteq \V_{\N_2}$.
 \end{lemma}%

\begin{proof}
Let $\V_{\N_1|_A} \not\subseteq \V_{\N_2|_A}$.  
 Then $\V_{\N_{1|A}} \cap \V_{\N_{2|A}} \subsetneq \V_{\N_{1|A}}$. In the proof of
Proposition 4.3 from \cite{Gross2017DistinguishingPN}, 
it is shown that if $\V_{\N_{1|A}} \cap \V_{\N_{2|A}} \subsetneq \V_{\N_{1|A}}$, then there exists a polynomial 
invariant $f$ contained in 
$\mathcal{I}_{N_2} \setminus \mathcal{I}_{N_1}$, which implies
that $\mathcal{I}_{N_2} \not \subseteq \mathcal{I}_{N_1}$, 
and so $\V_{\N_1} \not \subseteq \V_{\N_2}$.

\end{proof}


Lemma~\ref{lem:subnetVarieties} implies that
in order to prove Theorem \ref{thm: main} it will suffice to show that for any distinct triangle-free level-$1$ semi-directed networks $\N_1$ and $\N_2$, there either exists a set $A \subseteq \leafSet$ with $|A| = 4$ such that $N_1|_A$ and  $N_2|_A$ are distinguishable, or sets $A,B \subseteq \leafSet$ with $|A| = |B| = 4$ such that  $\V_{\N_1|_A} \not\subseteq \V_{\N_2|_A}$ and  $\V_{\N_1|_B} \not\supseteq \V_{\N_2|_B}$.


\section{Combinatorial properties of triangle-free level-1 semi-directed networks}\label{sec:combinatorial}


If $X \cup Y$ is a partition of $\leafSet$ such that $\N$ contains an $X-Y$ split, then denote by $N /X$ the network $\N|_{\{x\} \cup Y}$, for an arbitrary $x \in X$. Observe that the unrooted skeleton
of $\N/X$ does not depend on the choice of $x$.
Observe also that $r(\N) = r(\N/X) + r(\N/Y)$.

\begin{observation}\label{obs:splitSideDiffers}
 If $\N_1$ and $\N_2$ are distinct $n$-leaf semi-directed networks and $X-Y$ is a common split, then either $\N_1/X$ and $\N_2/X$ are distinct or $\N_1/Y$ and $\N_2/Y$ are distinct.
\end{observation}

The next lemma follows immediately from Lemma~\ref{lem:subnetVarieties} and the definition of $\N/X$.

\begin{lemma}\label{lem:splitSideInheritance}
 Let~$N_1$ and~$N_2$ be distinct~$n$-leaf semi-directed networks on~$\leafSet$.
 Suppose $X \cup Y$ is a partition of $\leafSet$ such that $\N_1$ and $\N_2$ both contain an $X-Y$ split. 
 If $\V_{\N_1/X} \not \subseteq \V_{\N_2/X}$ then $\V_{\N_1} \not\subseteq \V_{\N_2}$.
\end{lemma}

Let $\N$ be an~$n$-leaf triangle-free level-$1$ semi-directed network on $\leafSet$ and $C$ a cycle in $\N$.
Let $e_1, \dots, e_s$ be the cut-edges incident to $C$.
Then the \emph{partition induced by $C$} is the partition $X_1|\dots |X_s$ of $\leafSet$ such that $x \in X_i$ if and only if $x$ is separated from $C$ by $e_i$.
We say $X_i$ is \emph{below the reticulation vertex} if $e_i$ is the edge incident to the reticulation vertex in $C$.
If $X_i$ is below the reticulation vertex in $C$ then we also say that $x$ is \emph{below the reticulation vertex} for any $x \in X_i$.

We say a set of three or more leaves $\{x_1, \dots, x_t\}$ \emph{meet at a cycle $C$} if each leaf in $\{x_1, \dots, x_t\}$ appears in a different set of the partition induced by $C$.
We say that they \emph{induce} a cycle in $\N$ if $\N|_{\{x_1, \dots, x_t\}}$ is a $t$-cycle network. 
Note that if the set of leaves $\{x_1, \dots , x_t\}$ induce a cycle then they must meet at a cycle, but the converse does not hold unless one of $\{x_1 \dots, x_t\}$ is below the reticulation vertex.
As an example consider the network in Figure~\ref{fig:bigCentralCycle}: 
$\{a_1, a_2, a_3\}$ meet at the cycle $C_1$ but do not induce a cycle, whereas $\{x, a_1, a_2, a_3\}$ also induce a cycle.

Observe that if $\{x_1, \dots ,x_t\}$ ($t \geq 3$) meet at a cycle, then they meet in exactly one cycle in $N$, i.e., this cycle is unique in~$N$.
Denote this cycle by $C_{\N}(x_1,\dots, x_t)$.
(Note that  $C_{\N}(x_1,\dots, x_t)$ is not well-defined if $\{x_1, \dots ,x_t\}$ do not all meet at a cycle.)


Let~$N_1$ and~$N_2$ be distinct $n$-leaf triangle-free level-$1$ semi-directed networks on $\leafSet$, and let~$C_1$ be a cycle in~$N_1$ that induces a partition 
$A_1|\dots|A_s|X'$ with $X'$ below the reticulation vertex.
Let~$C_2$ be a cycle in~$N_2$ that induces a partition $B_1|\dots|B_t|X'$, with~$X'$ below the reticulation vertex.
We say that~$C_2$ \emph{refines}~$C_1$ if~$B_1|\dots |B_t$ is a refinement of $A_1|\dots|A_s$, i.e., if~$\bigcup_{i=1}^s A_i = \bigcup_{j=1}^t B_j$ and every pair of leaves~$a,b$ that are contained in different sets in~$A_1|\dots|A_s$ also appear in different sets in $B_1|\dots|B_t$.
See Figure~\ref{fig:CycleRefinement}.

\begin{figure}
    \centering
    \includegraphics[width=\textwidth]{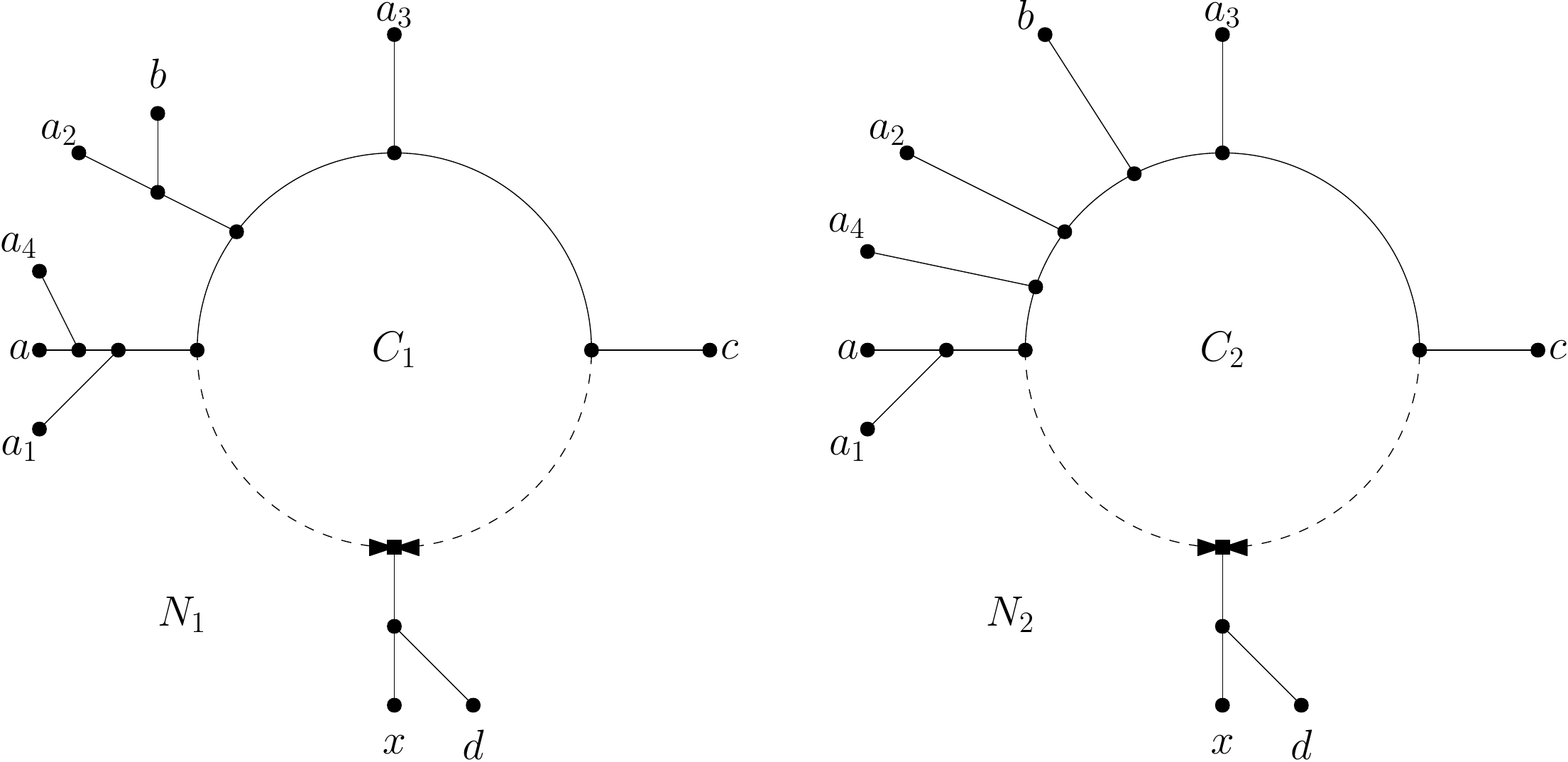}
    \caption{Two triangle-free level-1 semi-directed networks~$N_1$ and~$N_2$ on taxa set~$\{a,a_1,a_2,a_3,a_4,b,c,d,x\}$.
    The cycle~$C_1$ in~$N_1$ induces a partition~$\{a,a_1,a_4\}|\{a_2,b\}|\{a_3\}|\{c\}|\{x,d\}$ and the cycle~$C_2$ in~$N_2$ induces a partition~$\{a,a_1\}|\{a_4\}|\{a_2\}|\{b\}|\{a_3\}|\{c\}|\{x,d\}$. The cycle~$C_2$ refines~$C_1$.}
    \label{fig:CycleRefinement}
\end{figure}

We recall a combinatorial result from~\cite{banos2019} on four-leaf induced cycles. 
We state the result using notation and terms from this paper.

\begin{lemma}[Lemmas 12 and 13 of \cite{banos2019}]\label{lem:3leaves}
 Let~$N$ be an~$n$-leaf triangle-free level-1 semi-directed network on~$\leafSet$.
 If two distinct subsets of four leaves induce a 4-cycle, where three leaves in the two sets are the same, then the five leaves (the union of the two sets) meet at the same cycle.
 In other words, let~$a,b,c,d,e\in \leafSet$ be leaves of~$N$ such that~$N|_{\{a,b,c,d\}}$ and~$N|_{\{a,b,c,e\}}$ are both 4-cycle networks.
 Then~$\{a,b,c,d\}$ and~$\{a,b,c,e\}$ meet at the same cycle.
\end{lemma}

\begin{lemma}\label{lem:4CycleSupersetImpliesRefinedCycles}
Let $\N_1$ and $\N_2$ be distinct~$n$-leaf triangle-free level-$1$ semi-directed
networks on $\leafSet$.
Suppose that for any $a,b,c,d \in \leafSet$, if $\N_1|_{\{a,b,c,d\}}$ is a 4-cycle, then  $\N_2|_{\{a,b,c,d\}} =  \N_1|_{\{a,b,c,d\}}$.
Then every cycle in~$N_1$ is refined by a cycle in~$N_2$.
\end{lemma}
\begin{proof}
Let~$C_1$ be a cycle in~$N_1$ that induces a partition $A_1|\dots|A_s|X'$ with $X'$ below the reticulation vertex.
Choose any $a_1 \in A_1, a_2 \in A_2, a_3 \in A_3, x \in X'$. As $\N_1|_{\{a_1,a_2,a_3,x\}}$ is a $4$-cycle,  $\N_2|_{\{a_1,a_2,a_3,x\}}$ is the same $4$-cycle. So let $C_2 = C_{\N_2}(a_1,a_2,a_3,x)$.
We claim that $C_2$ is the desired cycle of~$N_2$ that refines~$C_1$.

To see this, first consider any $a \in A_h, b \in A_i, c \in A_j, d \in X'$ where $1 \leq h < i < j \leq s$.
Then $a,b,c,d$ all meet at $C_1$ and so $C_{\N_1}(a,b,c,d)$ is well-defined. Since $i,j>1$, we can replace $a$ with $a_1$ and have that the set of leaves $\{a_1, b,c,d\}$ also meet at $C_1$. By similar arguments, we also have that $\{a_1, a_2,c,d\}$  meet at $C_1$ and $\{a_1, a_2,a_3,d\}$  meet at $C_1$.
Moreover each of these sets of $4$ leaves induces a cycle in $N_1$ (as $d$ is below the reticulation vertex in $C_1$), and so also induce a cycle in $N_2$. Thus we have that $\N_2|_{\{a,b,c,d\}}$, $\N_2|_{\{a_1,b,c,d\}}$, $\N_2|_{\{a_1,a_2,c,d\}}$, $\N_2|_{\{a_1,a_2,a_3,d\}}$ are all $4$-cycles, and in particular $C_{\N_2}(a,b,c,d)$, $C_{\N_2}(a_1,b,c,d)$, $C_{\N_2}(a_1,a_2,c,d)$, $C_{\N_2}(a_1,a_2,a_3,d)$ are all well-defined. (See Figure~\ref{fig:cycleRefinementProof}.)
By Lemma~\ref{lem:3leaves}, we must have that $C_{\N_2}(a,b,c,d) =  C_{\N_2}(a_1,b,c,d) = C_{\N_2}(a_1,a_2,c,d) = C_{\N_2}(a_1,a_2,a_3,d) = C_{\N_2}(a_1,a_2,a_3,x) = C_2$.






We thus have that for  $a \in A_h, b \in A_i, c \in A_j, d \in X'$ with $h < i < j$, the set of leaves $\{a,b,c,d\}$ all meet at $C_2$.

Now consider any two leaves $a',b'$ such that $a',b'$ appear in different sets in $A_1|\dots|A_s|X'$.
By choosing additional leaves $c',d'$ from other sets, such that one of $a',b',c',d'$ is in $X'$,
we have that $C_{\N_2}(a',b',c',d') = C_{\N_2}(a,b,c,d)$ where $a\in A_h, b\in A_i,c\in A_j, d \in X'$, for some $h < i <j$. Then by the above we have that $C_{\N_2}(a',b',c',d') = C_2$. In particular, $a',b'$ appear in different sets in the partition induced by $C_2$.
This implies that the partition induced by $C_2$ is a refinement of the partition induced by $C_1$. Moreover, observe that $a'$ is below the reticulation vertex in $C_2$ if and only if $a' \in X'$ (since the only element of $\{a,b,c,d\}$ below the reticulation vertex in $C_2$ is the one from $X'$).
Thus, the partition induced by $C_2$ is $B_1|\dots|B_t|X'$ with $X'$ below the
  reticulation and $B_1|\dots |B_t$  a refinement of $A_1|\dots|A_s$.
  Therefore,~$C_2$ refines~$C_1$.
  \qed
\end{proof}

 \begin{figure}
    \centering
    \begin{subfigure}[t]{0.4\textwidth}
        \includegraphics[width=\textwidth]{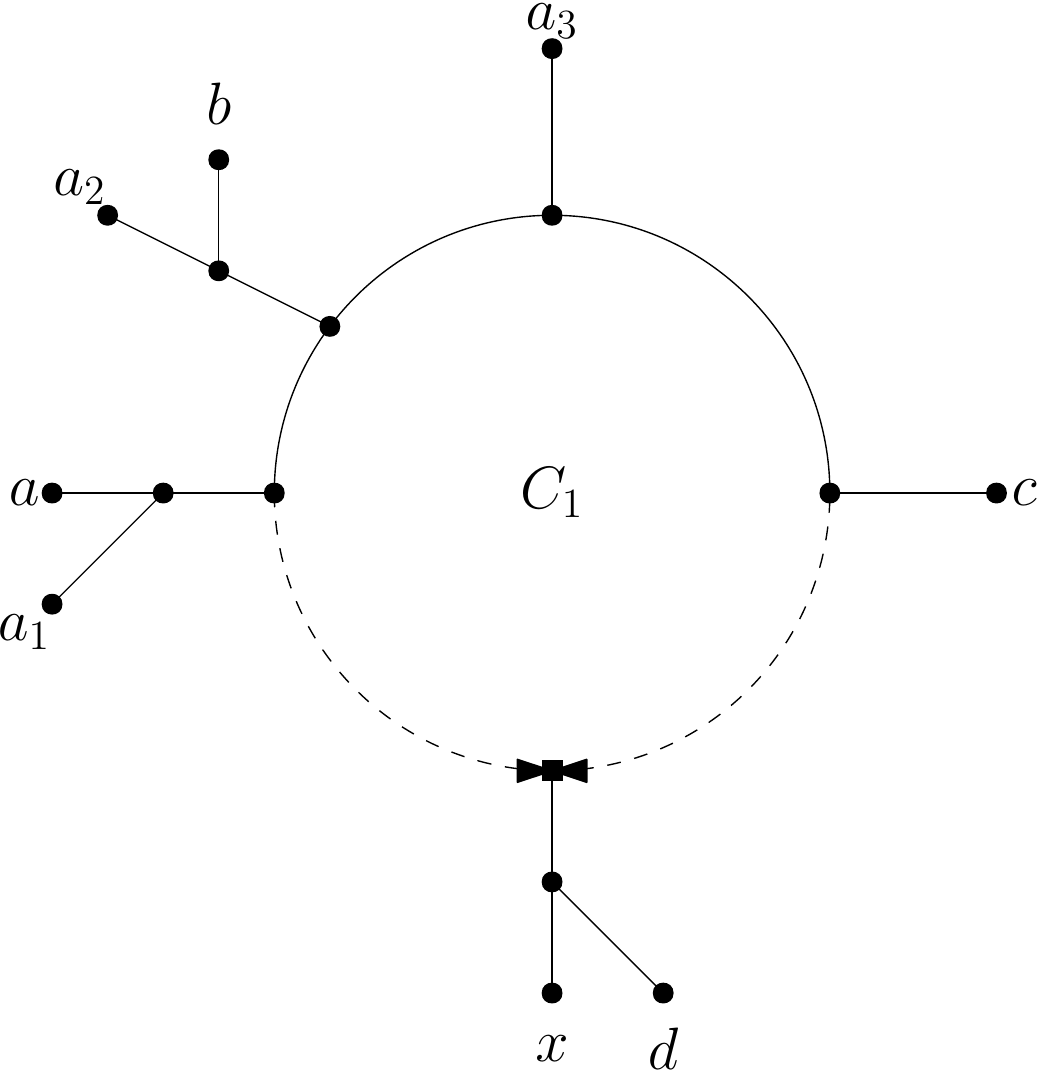}
        \caption{The cycle $C_1$.}\label{fig:bigCentralCycle}
    \end{subfigure}\hfill
    \begin{subfigure}[t]{0.4\textwidth}
        \includegraphics[width=0.8\textwidth]{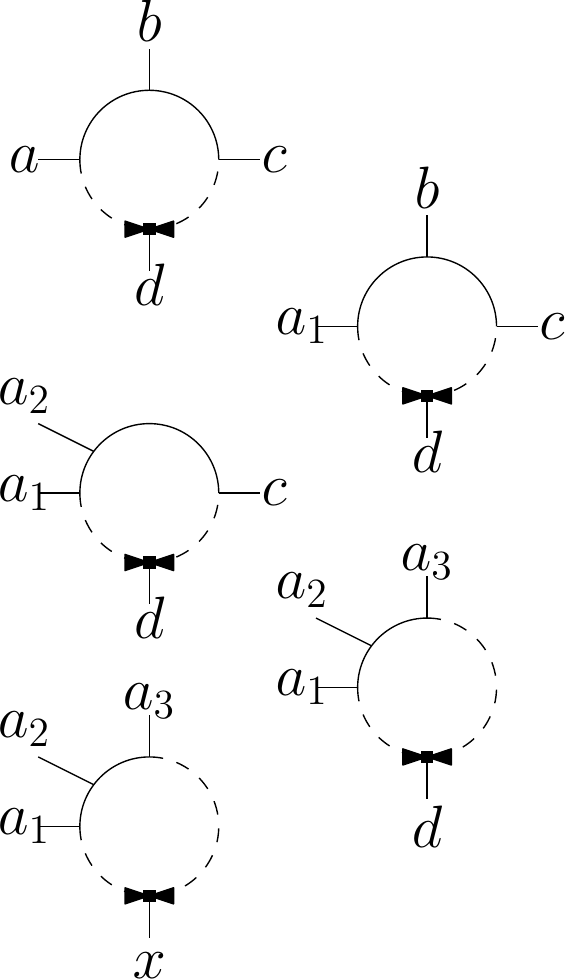}
        \caption{Some $4$-cycles that appear in $\N_1$ and thus in $\N_2$.}
    \end{subfigure}
    \caption{Illustration of part of the proof of Lemma~\ref{lem:4CycleSupersetImpliesRefinedCycles}. 
On the left we have an example of some leaves joining a cycle $C_1$ in $\N_1$, such that $\{a,b,c,d\}$ all meet at $C_1$ with $d$ below the reticulation vertex, and $\{a_1,a_2,a_3,x\}$ all meet at $C_1$ with $x$ below the reticulation vertex. The cycles on the right are all induced $4$-cycles in $\N_1$, and therefore by assumption are also induced $4$-cycles in $\N_2$. As the sets $\{a,b,c,d\}$ and  $\{a_1,b,c,d\}$ differ by only 1 element, they must meet at the same cycle in $\N_2$. By repeating a similar argument, we can show that $\{a,b,c,d\}$ and $\{a_1,a_2,a_3,x\}$ meet at the same cycle in $\N_2$.}
\end{figure}\label{fig:cycleRefinementProof}

 \begin{lemma}\label{lem:CommonORConflict}
 Suppose that every cycle in $\N_1$ is refined by a cycle in $\N_2$.
 If~$\N_2$ has a non-trivial split,
 then either $\N_1, \N_2$ share a non-trivial common split or they have conflicting splits.
 \end{lemma}
 \begin{proof}
 Let $A-B$ be a non-trivial split in $\N_2$.
 Fix an arbitrary $b \in B$, and take the edge $e$ in $\N_1$ furthest from $b$ such that $e$ separates $b$ from $A$.
 If $e$ separates $A$ from $B$, then $A-B$ is a non-trivial common split and we are done.
 
 Otherwise, let $u$ be the vertex in $e$ nearer to $A$.
 If $u$ is on a cycle, then denote this cycle by $C_1$.
 Let $X_1|\dots|X_s$ be the partition induced by $C_1$, noting by construction that $X_i \cap A = \emptyset$ for the set $X_i$ containing $b$ (since $X_i$ is the set of leaves reachable from $C$ via $e$).
 If $X_j \supseteq A$ for any $j$, then the corresponding edge $e_j$ leaving $C$ is an edge that is further away from $b$ than $e$ and which separates $A$ from $b$, contradicting the choice of $e$.
 So we may assume that the partition $X_1|\dots|X_s$ must \emph{subdivide} $A$ (that is, $A$ has non-empty intersection with at least two sets $X_j, X_h$).
 Furthermore  $X_1|\dots|X_s$ must subdivide $B$, as otherwise the set $X_i$ (which contains $b$) contains all of $B$ and also none of $A$,
 which would imply that $A-B$ is a common split.
 So $C_1$ is a cycle in $\N_1$ whose induced partition subdivides both $A$ and $B$. As every cycle in $\N_1$ is refined by a cycle in $\N_2$, this implies that some cycle in $\N_2$ also subdivides both $A$ and $B$. But this contradicts the fact that $\N_2$ contains an $A-B$ split.
 (See Figure~\ref{fig:claim1cycle}.)
 
If $u$ is not on a cycle, let $f$ and $g$ be the other edges incident to $u$. By choice of $e$, neither $f$ nor $g$ can separate $A$ from $b$. Thus there is at least one element $a \in A$ reachable from $u$ via $f$, and at least one element $a' \in A$ reachable from $u$ via $g$.
As $e$ does not separate $A$ from $B$, there is at least one $b' \in B$ that is reachable from $u$ via either $f$ or $g$, say (without loss of generality) $f$.
Then let $X-Y$ be the split induced by $f$, with $Y$ the set containing $b$.
Observe that $a,b' \in X$ while $a', b \in Y$.
Thus we have that  $X \cap A, X \cap B,  Y \cap A,  Y\cap B$ are all non-empty, and so $\N_1$ and $\N_2$ have conflicting splits.
 (See Figure~\ref{fig:claim1branch}.)
 \qed
 \end{proof}
 
  \begin{figure}
    \centering
    \begin{subfigure}[b]{0.4\textwidth}
        \includegraphics[width=\textwidth]{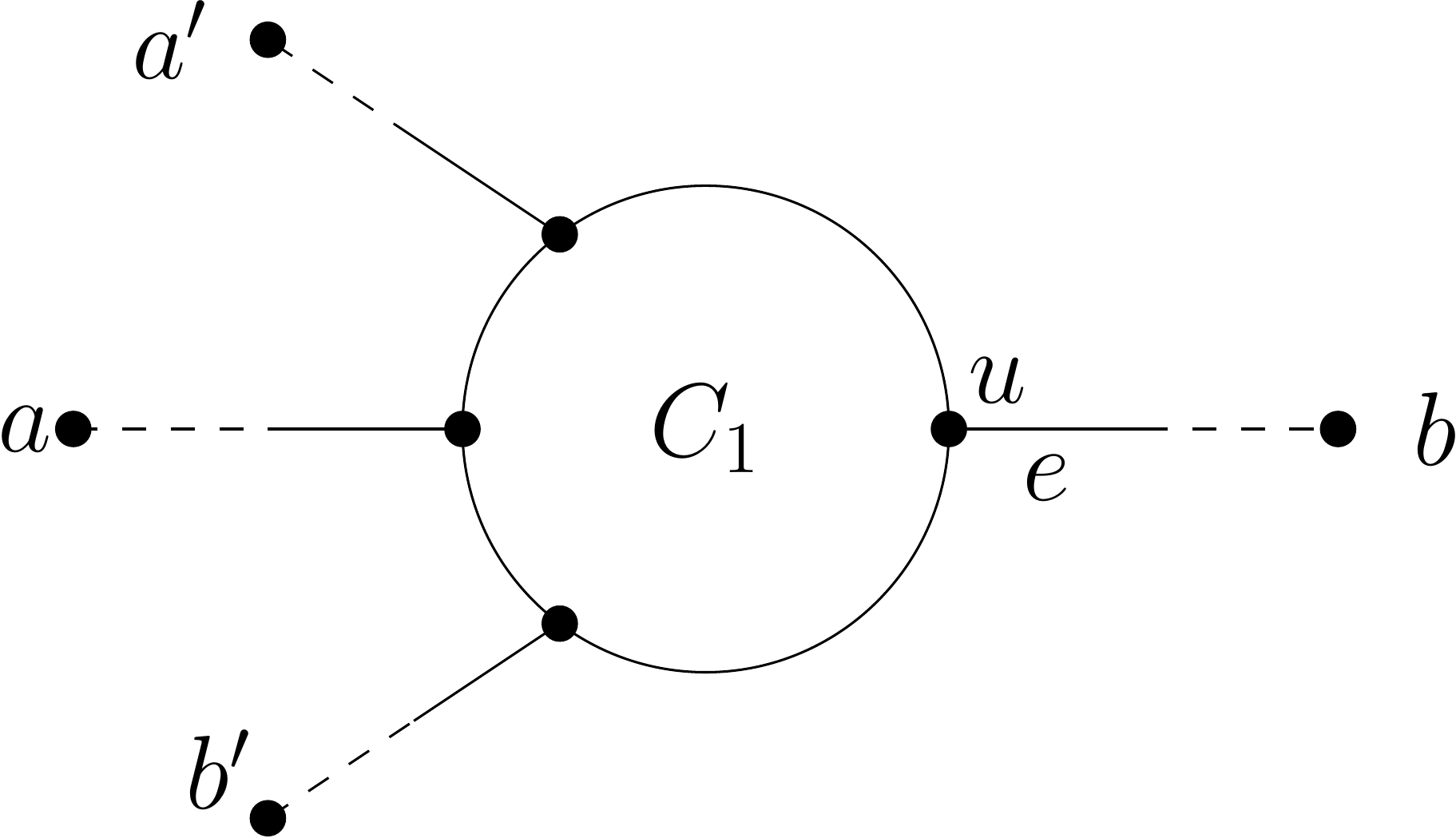}
        \caption{Case when $u$ is on a cycle.}
        \label{fig:claim1cycle}
    \end{subfigure}
    \begin{subfigure}[b]{0.4\textwidth}
        \includegraphics[width=0.8\textwidth]{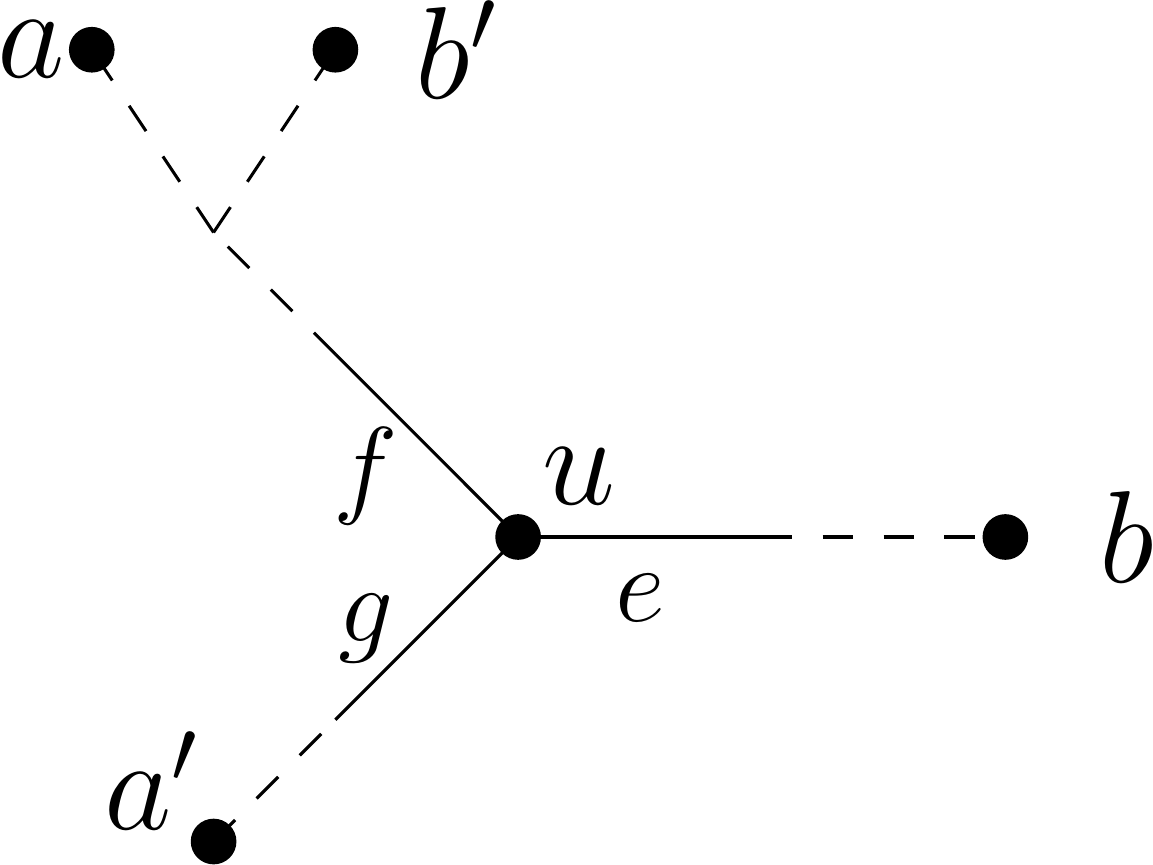}
        \caption{Case when $u$ is not on a cycle.}
        \label{fig:claim1branch}
    \end{subfigure}
    \caption{Illustration of~$N_1$ in the proof of Lemma~\ref{lem:CommonORConflict}.}
\end{figure}\label{fig:claim1proof}

%

\section{Distinguishability of triangle-free level-\texorpdfstring{$1$}{1} networks}
\label{sec: inductive proof}

%
%
Theorem~\ref{thm: main} follows as a corollary of the next lemma.

\begin{lemma}\label{lem:inductiveProof}
 Let $\N_1$ and $\N_2$ be distinct~$n$-leaf triangle-free level-1 semi-directed networks on~$\leafSet$ and $r(\N_1)\geq r(\N_2)$. 
 Then $\V_{\N_1} \not\subseteq \V_{\N_2}$ under the JC, K2P, and K3P constraints.
\end{lemma}
\begin{proof}
 We prove the claim by induction on $n=|\leafSet|$, the number of leaves in $\N_1$ and $\N_2$.
 For the base case, if $n \leq 4$ then either $r(\N_1) = 0$ or $r(\N_1) = 1$.
 If $r(\N_1) = 0$, then $N_1$ and $N_2$ are both trees.
If $r(\N_1) = 1$ and $r(\N_2) = 1$, 
then $N_1$ and $N_2$ are both 4-cycles. 
If $r(\N_1) = 1$ and $r(\N_2) = 0$, then 
$\N_1$ is a $4$-cycle network and $\N_2$ is a tree. For each of these cases, by Lemma~\ref{lem:buildingBlocks}, it follows that
$\V_{\N_1} \not\subseteq \V_{\N_2}$.
Note that we must have~$r(N_1)\le 1$ and~$r(N_2)\le 1$ as these networks are triangle-free. Thus, this covers all cases for~$n\le 4$.
 
 So now assume that $n > 4$ and that the claim is true for all smaller values of $n$.
 We first show that we may assume 
 that any set of 4 leaves that induces a 4-cycle in $\N_1$ induces the same 4-cycle in $\N_2$.
 Indeed, suppose this is not the case, and consider some arbitrary $A \subseteq \leafSet$ with $|A| = 4$ such that $\N_1|_A$ is a $4$-cycle but $\N_2|_A$ is not the same $4$-cycle.
 If $\N_2|_A$ is a \emph{different} $4$-cycle or a double-triangle, then by Lemma~\ref{lem:buildingBlocks}, $\N_1|_A$ and $\N_2|_A$ are distinguishable (and in particular,  $\V_{\N_1|_A} \not\subseteq \V_{\N_2|_A}$).  
 Otherwise, $\N_2|_A$ is either a tree or a $3$-cycle network, and Lemma~\ref{lem:buildingBlocks} implies that $\V_{\N_1|_A} \not\subseteq \V_{\N_2|_A}$.
In either case, $\V_{\N_1|_A} \not\subseteq \V_{\N_2|_A}$ and hence, by Lemma~\ref{lem:subnetVarieties}, we have
$\V_{\N_1} \not\subseteq \V_{\N_2}$.
 
 
 So we may now assume that any set of 4 leaves that induces a 4-cycle in $\N_1$ induces the same 4-cycle in $\N_2$.
 By Lemma~\ref{lem:4CycleSupersetImpliesRefinedCycles}, this implies that every cycle in~$N_1$ is refined by a cycle in~$N_2$.
 By Lemma~\ref{lem:CommonORConflict},~$\N_1$ and~$\N_2$ must have either a non-trivial common split or conflicting splits, or~$\N_2$ must have no non-trivial split.
 It remains to complete the proof in these three cases.

 
 Firstly, if $\N_1,\N_2$ have conflicting splits, then by Lemma~\ref{lem:conflictingSplits} we have
 $\V_{\N_1} \not\subseteq \V_{\N_2}$, as required.
 
 
 Secondly, suppose that $X-Y$ is a non-trivial common split, and consider $\N_1/X$ $\N_2/X$, $\N_1/Y$, $\N_2/Y$ as defined in the beginning of Section~\ref{sec:combinatorial}. 
 Since $|X|, |Y| \geq 2$, each of these networks has fewer than $n$ leaves.
 Thus by the induction hypothesis, if $\N_1/X, \N_2/X$ are distinct and $r(\N_1/X) \geq r(\N_2/X)$, then $\V_{\N_1/X} \not\subseteq \V_{\N_2/X}$, from which it follows that $\V_{\N_1} \not\subseteq \V_{\N_2}$.
 A similar argument holds if  $\N_1/Y, \N_2/Y$ are distinct and $r(\N_1/Y) \geq r(\N_2/Y)$.
 But at least one of these cases must hold.
Indeed, since $r(\N_1/X) + r(\N_1/Y) = r(\N_1) \geq r(\N_2) = r(\N_2/X) + r(\N_2/Y)$, it must hold that  $r(\N_1/X) > r(\N_2/X)$, $r(\N_1/Y) > r(\N_2/Y)$ or $r(\N_1/X) = r(\N_2/X)$ and $r(\N_1/Y) = r(\N_2/Y)$.
If $r(\N_1/X) > r(\N_2/X)$ (or $r(\N_1/Y) > r(\N_2/Y)$) then those networks are clearly distinct. 
Otherwise we have $r(\N_1/X) = r(\N_2/X)$ and $r(\N_1/Y) = r(\N_2/Y)$. We must have that $\N_1/X, \N_2/X$ are distinct or $\N_1/Y, \N_2/Y$ are distinct, since $\N_1$ and $\N_2$ are distinct.
Thus we either have that $\N_1/X, \N_2/X$ are distinct and $r(\N_1/X) \geq r(\N_2/X)$, or $\N_1/Y, \N_2/Y$ are distinct and $r(\N_1/Y) \geq r(\N_2/Y)$.
In either case we have  $\V_{\N_1} \not\subseteq \V_{\N_2}$, as required.
 
 Finally, suppose that~$\N_2$ has no non-trivial split. Then $\N_2$ is an $n$-cycle network, that is, $\N_2$ has a single cycle and every leaf is incident to a vertex on the cycle.
 If $r(\N_1) =1$, then $\N_1$ and $\N_2$ are both networks with exactly one cycle of length at least four. It then follows from 
 Theorem~\ref{thm: oneRetic},
 together with Proposition~\ref{prop: identifiability}, that~$\N_1$ and~$\N_2$ are distinguishable (and, in particular,  $\V_{\N_1} \not\subseteq \V_{\N_2}$). If on the other hand~$r(\N_1) \geq 2$, then consider two cycles~$C_1$ and~$C_2$ in~$\N_1$, with~$X_1'$ the subset of~$\leafSet$ below the reticulation in~$C_1$, and~$X_2'$ the subset of~$\leafSet$ below the reticulation in~$C_2$. Since~$C_1$ and ~$C_2$ are different cycles, $X_1' \neq X_2'$. But then this contradicts the fact that every cycle in~$N_1$ is refined by a cycle in~$N_2$, as the single cycle in~$\N_2$ would have to have both~$X_1'$ and $X_2'$ as the set of leaves below the reticulation. Thus in all cases we have either a contradiction or~$\V_{\N_1} \not\subseteq \V_{\N_2}$, which completes the proof of Lemma~\ref{lem:inductiveProof}.
 \qed
\end{proof}

We are now ready to prove Theorem~\ref{thm: main}, which we restate for convenience.
\setcounter{theorem}{1}
\begin{theorem}
The network parameter of a network-based Markov model under the Jukes-Cantor, Kimura 2-parameter, or Kimura 3-parameter constraints is generically identifiable with respect to the class of models where the network parameter is an $n$-leaf triangle-free, level-1 semi-directed network with $r \geq 0$ reticulation vertices.
\end{theorem}

\begin{proof}
Let $\{\mathcal M_N\}_{N\in \mathcal N}$ be a class of triangle-free, level-1 network models with a fixed number of reticulation vertices. 
Let~$N_1,N_2\in \mathcal N$ be distinct~$n$-leaf triangle-free level-$1$ semi-directed
networks on~$\leafSet$ with~$r(N_1) = r(N_2)$.
By invoking Lemma~\ref{lem:inductiveProof} twice, we have~$\V_{\N_1} \not\subseteq \V_{\N_2}$
and~$\V_{\N_2} \not\subseteq \V_{\N_1}$ under the JC, K2P, and K3P constraints.
By definition,~$N_1$ and~$N_2$ are distinguishable; as~$N_1$ and~$N_2$ were chosen arbitrarily from~$\mathcal N$, it follows that the semi-directed network parameter of~$\{\mathcal M_N\}_{N\in \mathcal N}$ is generically identifiable under the JC, K2P, and K3P constraints.
\qed
\end{proof}

\section{Discussion}

We have shown that triangle-free level-1 semi-directed networks are generically identifiable under the Jukes--Cantor, Kimura 2-parameter, and Kimura 3-parameter constraints. This means that, given a long enough multiple sequence alignment that evolved on a network of this class under one of these models, this network is, with high probability, the only network from the class that coincides with the given data. Roughly speaking, this means that the data provide sufficient information to reconstruct the network. To prove this result, we employed a blend of algebraic and combinatorial methods to show
that any pair of networks are geometrically distinguishable.

Previously, it had been shown that networks cannot be identified from certain substructures. 
For example, networks cannot be inferred from their displayed trees
since more than one network can display the same set of trees \citep{Gambette2012,Pardi2015}. 
Similarly, a network cannot in general be reconstructed from its collection of proper subnetworks, since two distinct networks can have exactly the same set of proper subnetworks~\citep{Huber2015}.
Nevertheless, for certain restricted network classes it has been shown that those networks can be uniquely reconstructed from their subnetworks~\citep{huber2013encoding,quarnets,van2014trinets,leonie}. These proofs are related to our combinatorial results, in that our proof strategy for showing network distinguishability involved careful examination of induced 4-leaf subnetworks.
However, there are some fundamental differences that prevent directly using known results on building networks from subnetworks. Firstly, the existing results focus either on directed (e.g.~\cite{van2014trinets}) or on undirected (e.g.~\cite{van2018leaf}) networks. Our results, as well as the ones in~\cite{allman2019,banos2019,quarnets},
provide the first combinatorial results on semi-directed networks.
The main obstacle, however, was that not all 4-leaf level-1 semi-directed networks are distinguishable under the considered models. Hence, two networks can be indistinguishable even if the sets of induced subnetworks are distinct. Consequently, we had a severely restricted set of building blocks available, requiring a combination of combinatorial and algebraic techniques.

On the algebraic front, the computations reveal differences between the relationships between the network ideals under the JC constraints and the relationships between the ideals under the K2P and K3P constraints that would be interesting directions for further exploration.   In \cite{hollering2020Identifiability}, the authors remark that the phenomenon observed in \cite{Gross2017DistinguishingPN} under the JC constraints, 
where each triangle network variety is contained within several of the 4-cycle network varieties, does not occur under the K2P and K3P constraints.
In other words, under the K2P and K3P constraints, 4-cycle networks and triangle networks are distinguishable.  In our computations for this paper, we noticed another phenomenon that seems to only hold for JC constraints.
In particular, under the
JC constraints, the ideals for the double-triangle
networks and the 4-cycle networks are of the same dimension and are all distinct. This is somewhat surprising as one
might expect the additional reticulation vertex and associated reticulation parameters 
of the double-triangle network to increase the dimension
of the model. Our numerical computations suggest
that this is another unique feature
of the JC constraints. However, establishing this result
rigorously may require other methods, since we were unable to compute full generating sets for the vanishing ideals of the networks under the K2P and K3P constraints. 
Additionally, from an algebraic perspective, we note that adapting the random search strategy described in \cite{hollering2020Identifiability} is what allowed us to find candidate subsets of variables 
for locating the necessary invariants to establish our main result. Something similar will likely need to be employed if these results are to be extended to other families of
networks. It would be interesting to understand the relative
computational costs once a candidate subset of variables is found, of either computing invariants in a subring of the original variables as we did, or of computing 
the linear matroid of the Jacobian with symbolic parameters as was done in \cite{hollering2020Identifiability}.

Finally, a major open problem, which is the larger setting for this paper, is to determine whether  generic identifiability results such as these can be extended to higher level networks.  
We expect finding the necessary invariants for the  increased number of non-unique induced 4-leaf subnetworks will be challenging.
Furthermore, the complexity of the combinatorial part of the proof will explode for higher levels. This question is open not only for the group-based models studied in this paper, but also for the general Markov model, which has just started to be studied in the context of networks \cite{casanellas2020rank}.

\end{document}